\documentclass{amsart}

\usepackage{mathptmx}
\usepackage{amssymb, amsmath, amsthm}
\usepackage{mathrsfs}
\usepackage[active]{srcltx}
\usepackage{verbatim}

\usepackage{physics}
\usepackage[fleqn,tbtags]{mathtools}

\usepackage[colorlinks,linkcolor={blue},citecolor={blue},urlcolor={red},]{hyperref}

\usepackage[shortlabels]{enumitem}

\usepackage{newsymbol}

\let\mathcal \undefined
\def\mathcal{\mathscr}

\let\emptyset \undefined
\let\ge       \undefined
\let\le       \undefined
\newsymbol\le          1336  \let\leq\le
\newsymbol\ge          133E  \let\geq\ge
\newsymbol\emptyset    203F
\newsymbol\notle       230A
\newsymbol\notge       230B

\theoremstyle{plain}
\newtheorem{theorem}{Theorem}[section]
\theoremstyle{remark}

\theoremstyle{plain}

\newtheorem{lemma}[theorem]{Lemma}

\newtheorem{definition}[theorem]{Definition}

\numberwithin{equation}{section}

\def\N{{\mathbb N}}

\def\R{{\mathbb R}}
\def\C{{\mathbb C}}

\renewcommand{\P}{{\mathbb P}}

\def\om{\omega}
\def\Om{\Omega}

\newcommand{\ran}{\mathsf{R}}
\newcommand{\iprod}[2]{( #1|#2 )}

\newcommand{\calL}{\mathscr{L}}
\newcommand{\n}{\Vert}
\newcommand{\one}{{{\bf 1}}}

\newcommand{\s}{^\star}

\newcommand{\Dom}{{\mathsf{D}}}

\newcommand{\ov}{\overline}
\newcommand{\wt}{\widetilde}
\newcommand{\wh}{\widehat}

\newcommand{\ud}{\,{\rm d}}

\newcommand{\calF}{\mathscr{F}}
\newcommand{\calE}{\mathscr{E}}
\newcommand{\calP}{\mathscr{P}}
\newcommand{\calB}{\mathscr{B}}
\newcommand{\calM}{\mathscr{M}}

\newcommand{\calA}{\mathscr{A}}

\newcommand{\D}{\Dom}

\renewcommand{\H}{\mathscr{H}}
\newcommand{\HS}{\calL^2(H^2(\mathbb{D}))}
\renewcommand{\D}{\mathbb{D}}
\newcommand{\U}{\mathbb{U}}

\allowdisplaybreaks

\begin{document}

 \title{Thermal time as an unsharp observable}
\author{Jan van Neerven}

\address{Delft Institute of Applied Mathematics\\
Delft University of Technology\\P.O. Box 5031\\2600 GA Delft\\The Netherlands}
\email{J.M.A.M.vanNeerven@tudelft.nl}

\author{Pierre Portal}
\address{Pierre Portal
Mathematical Sciences Institute
\\
Australian National University
\\
Ngunnawal and Ngambri Country
\\
Canberra ACT 0200, Australia}

\email{pierre.portal@anu.edu.au}

\date{\today}


\thanks{Most of this paper was written while the first-name author visited the Mathematical Science Institute at ANU (Canberra) through the MSRVP programme. He would like to thank the staff at ANU for the kind and generous hospitality.}

\begin{abstract}
We show that the Connes--Rovelli thermal time associated with the quantum harmonic oscillator can be described as an (unsharp) observable, that is, as a positive operator valued measure. We furthermore present extensions of this result to the free massless relativistic particle in one dimension
and to a hypothetical physical system whose equilibrium state is given by the noncommutative integral. \end{abstract}

\maketitle

\section{Introduction}\label{sec:intro}

It was recognised by Bohr in the early days of quantum mechanics that
the formulas $E = h\nu$ and $\lambda = h/p$, with $h$ the Planck constant
suggest the uncertainty relations
$$ \Delta x\,\Delta p \ge \frac{\hbar}{2} \ \ \hbox{and} \ \ \Delta t \,\Delta E \ge  \frac{\hbar}{2}.$$
Whereas position, momentum, and energy could be given a firm mathematical footing
in von Neumann's mathematical foundations of quantum mechanics as {\em observables}, that is,
selfadjoint operators acting on a suitable Hilbert space (`energy' being replaced by a
Hamiltonian operator), this cannot be done for time. In fact, it is a famous observation of Pauli \cite{Pauli}
that if $\mathscr{H}$ (reserving the letter $H$ for Hilbert spaces) is a semi-bounded Hamiltonian operator,
there exists no selfadjoint operator $T$ covariant with the unitary group
generated by $i\mathscr{H}$, that is, it cannot be true that
$$ e^{it\mathscr{H}}Te^{-it\mathscr{H}} = T-t, \quad t\in\R.$$
The common view is that `time', rather than being an observable, should be regarded as a mere bookkeeping index.

The spectral theorem establishes a one-to-one correspondence between selfadjoint operators
on the one hand and projection-valued measures on the real line. Observing that the orthogonal
projections on a Hilbert space $H$ are precisely the extreme points of the set of all bounded operators
$T$ on $H$ satisfying the relation $$0\le T\le I,$$ where $A\le B$ means that $\iprod{Ah}{h}\le \iprod{Bh}{h}$ for all $h\in H$,
a natural generalisation of the notion of a projection-valued measure is obtained upon replacing orthogonal projections
by operators satisfying the above relation. The mathematical theory of the resulting {\em positive operator-valued measures} (POVM) had already been developed in the 1940s by Naimark, and its usefulness in quantum mechanics was first advocated by Ludwig \cite{Ludwig}. It was quickly realised that covariant POVMs could be associated with various time measurements,
such as arrival time, times of occurrence, screen time, and time of flight, and rigorous versions of the time-energy uncertainty relation could be proved for some of them \cite{BruFre2008, BruFre2002,  Busch2008, Kij74, SriVij, Werner}; see also the overviews in \cite{BGL, Muga}.
For the free falling particle with Hamiltonian
$$ \mathscr{H} = \frac{P^2}{2m} - mgQ,$$
with $P$ and $Q$ the position and momentum operators, respectively, it is even possible to specify a selfadjoint operator
covariant with $\mathscr{H}$, namely
$$ T = \frac{P}{mg}.$$
This operator measures time as the increase of momentum. The problem of finding a projection-valued measure (PVM)
covariant with the free-particle
Hamiltonian remained elusive, however. As Busch \cite[p. 95-96]{Busch2008} puts it,
\begin{quote}
``Seemingly obvious candidates of a time operator conjugate to the free-particle
Hamiltonian,
$$ \mathscr{H}_{\rm free} = \frac{P^2}{2m},$$
are given by suitably symmetrised expressions for the time-of-arrival variable
suggested by classical reasoning, for example:
$$ -\frac12 m(QP^{-1}+P^{-1}Q) \ \ \ \hbox{or} \ \ -mP^{1/2} QP^{1/2} .$$
While these expressions formally satisfy the canonical commutation relation,
they are not selfadjoint but only symmetric (on suitably defined dense domains
on which they actually coincide, see Sect. 10.4), and they do not possess
a selfadjoint extensions. Hence this intuitive approach does not lead to a time
observable in the usual sense of a selfadjoint operator conjugate to the free
Hamiltonian.''
\end{quote}
The
problem of constructing a POVM for this setting, covariant with time translations, has been addressed, for instance, in \cite{Gia} using the time-of-arrival variable $-\frac12 m(QP^{-1}+P^{-1}Q)$. See also the books \cite{Holevo, BLPY}, and, in particular, Example 17.5.\cite{BLPY}.
In a similar vein, a description of the four space-time coordinates of an event in terms of a POVM defined on the Minkowski space-time covariant with respect to the
Poincar\'e group has been proposed in \cite{Toller1999a, Toller1999b}, while ``relative'' (or ``relational'') versions of time covariant POVMs have been studied in \cite{LovMiy}.

The aim of the present paper is to address
the problem of finding a POVM associated with the Connes-Rovelli {\em thermal time}. This notion is introduced in Section \ref{sec:thermal}.
After recalling some mathematical preliminaries in Section \ref{sec:prelim},  we provide in Section \ref{sec:number-phase} a positive answer for case of the quantum harmonic oscillator. In Theorem \ref{thm:thermal-L}
we show that the `clock time POVM' associated with the quantum harmonic oscillator as constructed in \cite{GarWong, BGL} transforms covariantly (up to a scale factor $\beta$) with the modular flow canonically associated with the Gibbs state $e^{-\beta N}$, where $N$ is the number operator.
Key to the proof of Theorem \ref{thm:thermal-L} is that this Gibbs state is normal. In order to extend the Theorem \ref{thm:thermal-L} to states that are not necessarily normal, it turns out to be fruitful to relax the definition of thermal time beyond the normal case while retaining the algebraic properties of the modular flow. The resulting notion of time, which we will call {\em modular time}, is introduced in Section \ref{sec:modular-time}.

Our second main result, Theorem \ref{thm:thermal-D}, identifies a POVM for the modular time of a free massless relativistic particle in dimension one at equilibrium with a heat bath of inverse temperature $\beta$. The main difficulty to overcome is that the equilibrium state associated with the Hamiltonian $|D| = \sqrt{\Delta}$ fails to be normal.

In Section \ref{sec:NC} we consider the semifinite singular trace $\om$ associated with the noncommutative integral. This trace, also known as the Dixmier trace, is comprehensively covered in the monograph \cite{LSZ}. It has deep applications in index theory \cite{Car1, Car2, Car3}, and has fundamental connections with general relativity \cite{Kastler} and
the noncommutative geometry approach to the standard model \cite{CC}. The third main result of this paper, Theorem \ref{thm:thermal-Dixmier}, shows that this trace is associated with a POVM on an abstract Hardy space naturally associated with (the Weyl calculus associated with) $\om$.

We conclude the paper with some open problems.

\section{Thermal time}\label{sec:thermal}

The notion of thermal time was introduced by Rovelli \cite{Rov93}
as a tentative solution to the `problem of time' in generally covariant theories (such as general relativity).
Such theories have no preferred time, but only the `internal times' of objects such as clocks. For this reason it has been proven difficult to give a covariant formulation of thermodynamics, which requires a preferred time to define the notion of equilibrium. Rovelli observed that in classical Hamiltonian mechanics it is possible to reverse this reasoning: the statistical state of the system can be used to determine a preferred internal time variable. The time variable determined in this way is then called {\em thermal time}. Rovelli's {\em thermal time hypothesis} is the postulate that this principle is generally applicable and captures the `common sense time flow' experienced as internal time. In a companion paper \cite{Rov-CMB}, Rovelli showed that in a Robertson--Walker universe, the thermal time associated with the cosmic background radiation precisely recovers the Robertson--Walker cosmological time.

The notion of thermal time was extended to the quantum domain by Connes and Rovelli \cite{ConRov} by interpreting the well-known connection between Kubo-Martin-Schwinger equilibrium states (KMS states) and the Tomita--Takesaki theory: a normalised state $\om$ on a $C^*$-algebra $\mathscr{A}$ is a KMS state with respect the Tomita--Takesaki flow in the GNS representation induced by $\om$. It is a deep result of Connes \cite{Con73} that this flow is essentially unique, in the sense that the flows corresponding to two different normalised states are always the same up to conjugation by an inner automorphism (a mapping of the form $A \mapsto UAU^\star$ for some unitary operator $U$). In view of its connection with KMS states,  this makes it natural to {\em interpret the Tomita--Takesaki flow as the thermal time associated with the state $\om$.}

\section{Mathematical preliminaries}\label{sec:prelim}

In this section we discuss the mathematical preliminaries needed for our purposes.
Notation is standard and follows \cite{Nee} unless otherwise stated.

\subsection{Positive operator-valued measures}\label{subsec:NCL2}

We start with recalling the definition and some properties of  positive operator-valued measures (POVMs).
For more information on this topic we refer to \cite{Berberian, BGL, Holevo, Kraus, Landsman1998, Ludwig, Nee}; the reference \cite{Nee} contains complete proofs of all theorems stated below.
POVMs arise in a variety of applications, such as the theory of open quantum systems \cite{DaviesOpen},
quantum information theory \cite{Wilde}, and the theory of symmetric operators \cite{AkhGla2}.

Let $H$ be a complex Hilbert space with inner product $\iprod{\cdot}{\cdot}$ and norm $\n\cdot\n$.
We use the mathematician's convention that
inner products are linear in their first argument and conjugate-linear in their second argument. The Banach space of all bounded linear operators on $H$ is denoted by $\calL(H)$.

An {\em effect} is a bounded operator  $E\in \calL(H)$ satisfying  $0\le E\le I$. The set of effects on $H$ will be denoted by $\mathscr{E}(H)$. The set of its extreme points equals $\calP(H)$, the set of all orthogonal projection on $H$. This suggests generalising the standard definition of a projection-valued measure by replacing the role of orthogonal projections by effects. Accordingly,
a {\em positive op\-erator-valued measure} (POVM)
on a measurable space $(\Om,\calF)$ is defined to be a mapping $E: \calF\to \calE(H)$ assigning to every set $F\in \calF$ an effect $E_F\in \calE(H)$ with
the following properties:
\begin{enumerate}[label={\rm(\roman*)}, leftmargin=*]
 \item\label{it:POVM1} $E_\Om = I$;
 \item\label{it:POVM2}  for all $x\in H$ the mapping $F\mapsto \iprod{E_F x}{x}$
defines a measure $E_{x}$ on $(\Om,\calF)$.
\end{enumerate}
The measure defined by \ref{it:POVM2} is denoted by $E_{x}$. From
$$ E_x(\Om) = \iprod{E_\Om x}{x} = \iprod{x}{x} =\n x\n^2$$ we infer that every measure $E_x$ is finite.
It is easy to see that a POVM $E: \calF\to \calE(H)$ is a projection-valued measure
if and only if for all $F,F'\in \calF$ we have $$E_FE_{F'} = E_{F\cap F'}.$$

The convex set of all probability measures on a measurable space $(\Om,\calF)$ is denoted by $M_1^+(\Om)$.
The quantum analogue of this set is the set $\mathscr{S}(H)$ of all positive trace class operators $T$ on $H$ with unit trace;
in the physics literature such operators are referred to as {\em density matrices}.
The formula $$\om(S) = \tr(ST)$$ establishes a one-to-one correspondence between the operators $T\in \mathscr{S}$ and the set of all {\em normal states}, that is, bounded linear linear functionals $\om:\calL(H)  \to \C$ with the property that $\sum_{n\ge 1}\phi(P_n) = \phi(P)$ whenever $(P_n)_{n\ge 1}$ is a sequence of disjoint orthogonal projections in $H$ and $P$ is their least upper bound.

\begin{theorem}\label{thm:POVM-unsharp}
If $\Phi: \mathscr{S}(H) \to M_1^+(\Om)$ is a convexity preserving mapping, then there exists a unique POVM $E:\calF\to\calE(H)$ such that for all $T\in\mathscr{S}(H)$ we have
$$ (\Phi(T))(F) = \tr(E_F T), \quad F\in \calF\!.$$
\end{theorem}

This theorem establishes the operational role of POVMs as `unsharp observables': in that it transforms normal quantum states
to classical states.

If $J$ is an isometry from $H$ into another complex Hilbert space $\wt H$ and
$F\mapsto \wt P_F \in \calP(\wt H)$ is a projection-valued measure in $\wt H$, then
$$F\mapsto J^\star \wt P_F J \in \calE(H)$$ is a POVM in $H$.
A celebrated theorem of Naimark asserts that the converse is also true:

\begin{theorem}\label{thm:Naimark} Let $(\Om,\calF)$ be a measurable space and let
$E: \calF\to \calE(H)$ be a POVM. There exists a Hilbert space $\wt H$,
a projection-valued measure $\wt P: \calF\to \calP(\wt H)$, and an isometry $J:H\to \wt H$ such that for all $F\in \calF$ we have
$$ E_F = J^\star \wt P_F J.$$
\end{theorem}

For any POVM $E:\calF\to \calL(H)$ there exists a unique linear mapping $\Psi:B_{\rm b}(\Om)\to \calL(H)$
satisfying
$$ \Psi(\one_F) = Q_F, \quad F\in \calF\!,$$
  and
$$ \n \Psi(f)\n \le \n f\n_\infty, \quad f\in B_{\rm b}(\Om),$$
where $B_{\rm b}(\Om)$ is the Banach space of all bounded measurable functions on $\Om$ endowed with the supremum norm
$\n f\n_\infty = \sup_{\om\in\Om}|f(\om)|$.
This mapping $\Psi$ satisfies $$\Psi(f)^\star = \Psi(\ov f), \quad f\in B_{\rm b}(\Om).$$
In contrast to the corresponding functional calculus for PVMs, unless the POVM is a PVM, this calculus fails to be multiplicative. Nevertheless, in the same way as bounded normal (resp. selfadjoint, unitary) operators are in one-to-one correspondence with PVMs on $\C$ (resp. on $\R$, $\mathbb{T}$),  contractions are in one-to-one correspondence with POVMs on the unit circle $\mathbb{T}$ and admits the following representation:

\begin{theorem}\label{thm:contr-POVM} For every contraction $T\in\calL(H)$, there exists a unique POVM $E$ on the Borel
$\sigma$-algebra $\calB(\mathbb{T})$
 such that $$ T^n = \int_{-\pi}^\pi e^{in\theta} \ud E(\theta), \quad n\in\N.$$
If $E$ is a POVM with the above property, then $T$ is unitary if and only if $E$ is a projection-valued measure.
\end{theorem}

Intriguingly, the notion of POVM admits a
category-theoretic counterpart, in the sense that a Naimark-type theorem holds and
the usual notion of POVM is recovered in the category of Hilbert spaces
and bounded linear operators \cite{CoePaq}.

\subsection{The Tomita--Takesaki flow}

We can only give the briefest of descriptions here; for a full account the reader is referred to \cite{BraRob, Haa, Takesaki}, where the proofs of all statements can be found.

Let $H$ be a separable complex Hilbert space (in applications it usually is the representation Hilbert space arising from a GNS-like construction), let $\calM\subseteq\calL(H)$ be a von Neumann algebra, and let $\Om\in H$ be a {\em separating} and {\em cyclic} element for $\calM$. Recall that this means that whenever $A\in\calM$ satisfies $A\Om = 0$ it follows that $A=0$, and that the subspace $\calM\Om = \{A\Om :\ A\in\calM\}$ is dense in $H$.
The antilinear operator
$$ S_0: A\Om \mapsto A^\star \Om,$$
is well defined as a densely defined operator in $H$ with domain $\Dom(S_0) = \calM \Om$, and this operator is closable. Its closure will be denoted by $S$. By polar decomposition, there exist a unique antiunitary operator $J$ on $H$ and a unique positive selfadjoint operator $\Delta$ in $H$ such that
$$ S = J\Delta^{1/2}$$
with equal domains.
The operator $J$ and $\Delta$ are called the {\em modular conjugation} and the {\em modular} associated with the triple $(\calM,\tau,\Omega)$.
These operators have a number of interesting properties that are of no concern here. Important for us is that $\Delta$ induces
a one-parameter family of automorphisms of $\calM$ by
$$ \sigma_t \colon A \mapsto \Delta^{-it}A\Delta^{it}, \quad t\in\R,$$ the so-called
{\em modular flow associated with $\Om$}. The non-trivial point is that the right-hand side belongs to $\calM$ for all $A\in\calM$. Note that, following \cite{ConRov}, on the right-hand side we have changed the usual sign convention to enhance compatibility with the physical applications.

Let us now suppose that $\om$ is a {\em state} on a $C^\star$-algebra $\calA$ with unit $I$, that is,
$\om:\calA\to\C$ is a bounded linear functional on $\calA$ with the following two properties:
$$ \om(I) = 1 , \qquad \om(A^\star A)\ge 0 \ \ \hbox{for all} \ \ A\in\calA.$$
There is some redundancy in this definition, in that boundedness is a consequence of the other assumptions.
It is an easy consequence of the Cauchy--Schwarz inequality \cite[Lemma 2.3.10]{BraRob} that the set $\mathscr{N}_\om := \{A\in \calA \colon \om(A^\star A) = 0\}$ is a left ideal of $\calA$. On the quotient space $\calA/\mathscr{N}_\om$,
$\om$ induces an inner product $$\iprod{q_\om(A)}{q_\om(B)}_\om := \om(B^\star A),$$ where $q_\om: \calA \to \calA/\mathscr{N}_\om$ is the quotient mapping. Let $H_\om$ be the Hilbert space obtained by completing the quotient with respect to it,
and define $\pi_\om : \calA \to \calL(H_\om)$ by
$$ \pi_\om(A) q_\om(B) := q_\om(AB), \quad B\in \calL(H_\om).$$
Then $\pi_\om$ extends to a contractive $\star$-homomorphism from $\calA$ to $H_\om$, and this mapping is an isometry if and only if the state $\om$ is {\em faithful}, that is, $\om(A^\star A) =0$ implies $A=0$, or equivalently, if $\mathscr{N}_\om = \{0\}$.
The element $$\Om:= q_\om(I)$$ is cyclic for the range of $\pi_\om$; in particular, $\Om$ is cyclic for the von Neumann algebra $\calM_\om$ generated by the range of $\pi_\om$ in $\calL(H_\om)$. Note that
for all $A\in\calA$ we have $$ \om(A) = \iprod{\pi_\om(A)\Om}{\Om}_\om.$$
The pair $(H_\om,\pi_\om)$ is known as the {\em Gelfand--Naimark--Segal} (GNS) {\em representation} associated with $\om$.
If, in addition to being cyclic, $\Omega$ is also separating for $\calM_\om$, this places us in the setting of the Tomita--Takesaki theory, and the associated
modular flow on $\calM_\om$ will be called the {\em modular flow associated with the state $\om$}.

The connection with equilibrium thermodynamics is given by a theorem (see \cite[page 218]{Haa})
that states that if $\Om$ is separating for $\calM_\om$, the state $\om$ is a KMS state with respect to the modular flow associated with $\om$.
The precise definition of KMS states is unimportant for our purposes (for an exhaustive treatment see \cite{BraRob2}). What matters is that this notion extends the notion of a Gibbs state, in that every Gibbs state in KMS; and the notions are equivalent under suitable finiteness conditions. Recall that
if $\mathscr{H}$ is a Hamiltonian operator with the property that the bounded operator $e^{-\beta \mathscr{H}}$ is of trace class, then the
{\em Gibbs state at inverse temperature $\beta>0$ associated with $\mathscr{H}$} is the state $\om$ on $\calL(H^2(\mathbb{D}))$
given by $$\om(A) = \tr(T_\beta A),$$
where
$T_\beta:= {e^{-\beta \mathscr{H}}}/{\tr(e^{-\beta \mathscr{H}})}.$

\subsection{Noncommutative $L^2$-spaces}

The noncommutative spaces $L^2(\calM,\tau)$, where $\calM$ is a von Neumann algebra, are usually introduced for faithful semifinite normal traces $\tau$. Since in Section \ref{sec:NC} we will be dealing with general semifinite traces $\tau$, we will briefly go through the steps of defining $L^2(\calM,\tau)$ in that more general setting.

Let $\calM\subseteq\calL(H)$ be a von Neumann algebra over a given abstract separable Hilbert space $H$, and
let $\tau$ be a trace $\tau$ on $\calL(H)$, that is, $\tau$ is a mapping from the positive cone in $\calM$ to $[0,\infty]$ such that
the following conditions are satisfied:
\begin{enumerate}[\rm(i)]
 \item $\tau(cA) =c\tau(A)$ for all  $0\le A\in \calM$ and scalars $c\ge 0$;
 \item $\tau(A+B) =\tau(A)+\tau(B)$ for all $0\le A\in \calM$ and $0\le B\in \calM$;
 \item $\tau(A^\star A) = \tau(AA^\star)$ for all $A\in \calM$.
\end{enumerate}

Let
\begin{align*}
 \calM_\tau^{1,+} & := \{A\in\calM:\, A\ge 0, \ \tau(A) < \infty\}, \\
 \calM_\tau^2 & := \{A\in\calM:\, \tau(A^\star A) < \infty\}.
\end{align*}
The linear span  $\calM_\tau^{1}$ of $\calM_\tau^{1,+}$ is a $\star$-subalgebra of $\calM$ and
the restriction of $\tau$ to $\calM_\tau^{1,+}$ uniquely extends to a
positive linear functional on $\calM_\tau^{1}$.
Moreover, $\calM_\tau^2$ is a left ideal in $\calM$, and for any two $A,B\in \calM_\tau^2$ one has $B^\star A \in \calM_\tau^{1}$ and the {\em Cauchy--Schwarz inequality} holds:
$$ |\tau(B^\star A)|^2 \le \tau(A^\star A)\tau(B^\star B)$$
(see \cite[Lemma 5.1.2]{Ped}, which applies here since condition (iii) implies that
$\tr(UAU^\star) = \tr(A)$ if $U$ is unitary).
By polarisation (see the proof \cite[Proposition 5.2.2]{Ped}), for all  $A,B\in \calM_\tau^2$ we have
\begin{align*} \tau(AB) = \tau(BA).
\end{align*}

By the Cauchy--Schwarz inequality, the set
$$ \mathscr{N}_{\tau}:= \{A\in\calM_\tau^2: \ \tau(A^\star A) = 0\}$$
is a vector space. As in \cite[Lemma 5.1.2]{Ped}),
$ \mathscr{N}_{\tau}$ is a left ideal
in $\calM_\tau^2$. On the quotient space $\calM_\tau^2/\mathscr{N}_{\tau}$ we obtain a well defined inner product by putting
$$ \iprod{q_\tau(A)}{q_\tau(B)}_\tau := \tau(B^\star A), \quad A,B\in \calM_\tau^2,$$
where $q_\tau: \calM_\tau^2\to \calM_\tau^2/\mathscr{N}_{\tau}$ is the quotient mapping.
The Hilbert space  $H_{\tau}$ is defined as the completion of $\calM_\tau^2/\mathscr{N}_{\tau}$ with respect to this induced inner product.

\section{Thermal time for the quantum harmonic oscillator}\label{sec:number-phase}

As is well known, the (one-dimensional) quantum harmonic harmonic oscillator is described by the Hamiltonian $L = \frac12I+N$, where $N$ is the number operator and $\frac12$ is the ground state energy. In view of the trivial identity $e^{itL}Ae^{-itL} = e^{itN}Ae^{-itN}$, from the point of view of finding covariant POVMs it makes
no difference whether we work with $L$ or $N$. Since the number operator $N$ admits a natural presentation on the space $H^2(\mathbb{D})$ (which, incidentally, also suggests a natural approach to the problem of time for the free relativistic particle), we will work with $N$ instead of $L$.

The problem of finding a POVM covariant with the unitary group $(e^{itN})_{t\in\R}$ has already been solved in the physics literature \cite{BGL} (using ideas that go back to \cite{GarWong}), where this POVM is interpreted as describing `phase' as a (circle-valued) `unsharp observable'. As was already observed in this paper, since $L$ can be interpreted as a Hamiltonian this POVM can alternatively be interpreted as describing `(oscillator clock) time'. The exposition presented here follows the mathematically rigorous treatment of \cite[Chapter 15]{Nee}.

\paragraph{Thermal time as a POVM in $H^2(\mathbb{D})$}

Let $\mathbb{D}$ denote the open unit disc $\{z\in\C:\, |z|<1\}$ in the complex plane.
The space $H^2(\mathbb{D})$ is defined as the Hilbert space whose elements are the holomorphic functions $f:\mathbb{D}\to \C$ whose power series $f(z) = \sum_{n\in\N} c_n z^n$ satisfies
$$\n f\n_{H^2(\mathbb{D})}^2 = \sum_{n\in\N} |c_n|^2.$$ We consider the selfadjoint operator $N$
with domain
$$\Dom(N) = \Bigl\{f = \sum_{n\in\N} c_n z_n\in H^2(\mathbb{D}):\, \sum_{n\in\N} n^2|c_n|^2 <\infty\Bigr\},$$
given by
$$ N z_n =nz_n, \quad n\in\N,$$
where $z_n(z) := z^n$ as a function on $\D$.
Then $\sigma(N) = \N$ and the functions $z_n$ form an orthonormal basis of eigenvectors for $N$.

To define {\em phase}\index{phase} as a POVM on $\mathbb{T}$ we proceed as follows.
For Borel sets $B\subseteq \mathbb{T}$ we define the operator $E_B$ on $H^2(\mathbb{D})$
by $$ E_B = J^\star\circ P_B \circ J,$$
where $J:H^2(\mathbb{D}) \to L^2(\mathbb{T})$ is the isometry given by $$J: \sum_{n\in\N}c_n z_n\mapsto \sum_{n\in\N}c_n e_n,$$
$J^\star: L^2(\mathbb{T})\to H^2(\mathbb{D})$ is its adjoint, and $P_B f:=\one_B f$ for $f\in L^2(\mathbb{T})$.
It is clear that $0\le E_B\le I$, and accordingly the assignment $E: B\mapsto E_B$ defines a POVM on $\mathbb{T}$.
This construction goes back to \cite{GarWong}, where a selfadjoint operator $F$ on $H^2(\mathbb{D})$ is constructed
satisfying the Heisenberg commutation relation $$NF-FN = -iI$$ on a suitable dense domain of function in $H^2(\mathbb{D})$.
This operator is defined as the Toeplitz operator with symbol $\arg(z)$.
The problem with this construction, however, is that there the argument is not uniquely defined; one obtained different operators for different choices. As observed in \cite{Nee-Phase} the Heisenberg relation extends to the natural maximal domain of the commutator $FN- NF  = -iI$, but this domain fails to be dense with respect to the graph norm of $N$.
Moreover, as already observed in \cite{GarWong}, the (stronger) Weyl commutation relations
$$ e^{isN}e^{itF} = e^{-ist}e^{itF}e^{isN}$$
do not hold.
The POVM constructed above solves these problems in a satisfactory way,
in that it enjoys the following covariance property with respect to (cf. \cite[Section III.5]{BGL}):

\begin{theorem}\label{thm:quantum-phase}
The phase observable $E$ is {\em covariant} with respect to the unitary $C_0$-group generated by $iN$, that is,
for all Borel subsets $B\subseteq\mathbb{T}$ we have
\begin{align*}
e^{-itN}E_B e^{itN} = E_{e^{it}B}, \quad t\in \R,
\end{align*}
where $e^{it}B = \{e^{it}z: \, z\in B\}$ is the rotation of $B$ over $t$.
\end{theorem}

The relationship between the POVM $E$ and the operator $F$ is expressed by the relation
$$ F = \int_{\mathbb{T}}\arg(z) \ud E(z).$$
We refer the reader to \cite{Nee, Nee-Phase} for a proof and a fuller discussion of these matters.

Reinterpreted as `clock time' (cf. \cite{GarWong}), we show next that the POVM $E$ constructed in the previous subsection can be interpreted as a special instance of the Connes--Rovelli thermal time associated with the Gibbs state
$$ \om(A):= \tr (AT), \quad A\in \calL(H^2(\mathbb{D})),$$
where $$ T = e^{-\beta N}/\tr\,(e^{-\beta N}).$$
In this formula, the scalar $\beta$ is interpreted as the inverse temperature.
This involves showing that $E$ transforms covariantly with the modular flow associated with $\om$ by the Tomita--Takesaki theory, which is the content of Theorem \ref{thm:thermal-L} below. Apart from its own interest, the proof will serve as a template for our definition of modular time in the next section.

Following \cite{HBW} (see also \cite[Chapter V]{Haa}), rather than using the GNS representation associated with $\om$, we may use the unitarily equivalent representation $\pi$ of $ \calL(H^2(\mathbb{D}))$
on the Hilbert space $\HS$ of all Hilbert-Schmidt operators on $H^2(\mathbb{D})$ given by
$$ \pi(A)B:= AB, \quad A\in \calL(H^2(\mathbb{D})), \ B \in \HS.$$
In what follows we will write $A$ for $\pi(A)$.

Since $T$ is of trace class,  $\Omega:= T^{1/2}$ is Hilbert-Schmidt and therefore defines an element of $\HS$, and it is immediate that
$$ \om(A) = \iprod{A\Om}{\Om},\quad A\in \calL(H^2(\mathbb{D})),$$
where $\iprod{A}{B} = \tr\, (B^\star A)$ is the inner product of $\HS$.
As consequence of the unitary equivalence of this representation with the GNS representation, $\Om$ is
cyclic for $\pi$ (and hence for $\calM$, the von Neumann algebra $\calM$ over
$\HS$ generated by $\calL(H^2(\mathbb{D}))$, viewing operators in $\calL(H^2(\mathbb{D}))$ as operators acting on $\calL(\HS)$ by left multiplication) and separating for $\calM$. Moreover,
by \cite[Theorem 2.4.24]{BraRob}, $\calM = \calL(\HS)$.

On $\HS$ we consider the unitary operators $U(t)$, $t\in\R$, defined by
$$ U(t)S:= e^{i\beta tN}Se^{-i\beta tN}, \quad S\in \HS.$$
The family $(U(t))_{t\in\R}$ is easily seen to be a unitary $C_0$-group on $\HS$.
By Stone's theorem,
the generator of this group is of the from $i\mathscr{H}$ with $\mathscr{H}$ selfadjoint.
By the functional calculus of unbounded selfadjoint operators, the operator $\Delta:= e^{-\mathscr{H}}$ is
injective and selfadjoint (its domain being given by this calculus).
By the composition rule of the functional calculus, for $t\in \R$ we obtain
$$ \Delta^{-it} = (e^{-\mathscr{H}})^{-it} = e^{it \mathscr{H}} = U(t),$$
where in the last step we used the fact that the unitary group generated by $i\mathscr{H}$ also arises through exponentiation in the functional calculus. In particular, for all $A,B\in\calL(H^2(\mathbb{D}))$,
$$ \iprod{\Delta^{-it}A\Om }{B\Om}
= \om(B^\star e^{i\beta tN} Ae^{-i\beta tN}).$$
By spectral theory, for $h\in \Dom(\Delta^{1/2})$
the mapping $z\mapsto \Delta^{-iz}h$
is continuous on the closed strip $ \{\Im z \in [0,\frac12]\}$
and holomorphic on its interior. This, and standard properties of fractional powers, implies
that for all $A,B\in\calL(H^2(\mathbb{D}))$, both sides of the identity
$$ \iprod{\Delta^{-it}A\Om}{B\Om} = \om(B^\star e^{i\beta tN}Ae^{-i\beta tN}) =\tr\, (B^\star e^{i\beta tN}Ae^{-\beta(it+1)N}) $$
admit a continuous extension to the strip $\{\Im z \in [0,\beta]\}$ which is analytic on the interior,
and given by substituting $z$ for $t$.
By the edge of the wedge theorem (cf. \cite[Proposition 5.3.6]{BraRob2}) these extensions agree. It follows that for $\Re z\in [0,\frac12]$ we have $A\Om\in\Dom(\Delta^{1/2}) \subseteq\Dom(\Delta^{-iz})$,
and taking $z = \frac12i$ results in the identity
\begin{align*}
\iprod{\Delta^{1/2} A\Om}{B\Om} = \iprod{e^{-\frac12\beta \mathscr{H} } A\Om}{B\Om}
& = \tr\, (B^\star e^{-\frac12\beta N}Ae^{-\frac12\beta N})\\
&= \iprod{\Om A}{B\Om}.
\end{align*}
Since $\Om$ is cyclic for $\pi$ in $\HS$, this implies
$$ \Delta^{1/2} A\Om = \Om A.$$

Denoting by $J$ the antiunitary operator that sends an operator $B\in \HS$ to its adjoint $B^\star\in\HS$,
it follows that
$$ (J\Delta^{1/2}) (A\Om) = J(\Om A) = A\s \Om.$$
By uniqueness, this implies that $J$ and $\Delta$ are the modular conjugation and the modular associated with $\om$ in the representation $\pi$, respectively.

Denoting the associated modular flow as before by $ \sigma_t A := \Delta^{-it}A\Delta^{it}$, we have the following result:

\begin{theorem}\label{thm:thermal-L}
For all $t\in\R$ and Borel sets $B\subseteq \mathbb{T}$, the POVM $E$ satisfies $$\sigma_t (E_B)  = E_{e^{i\beta t}B}.$$
\end{theorem}
\begin{proof}
For all $A,B\in \calL(H^2(\mathbb{D}))$ we have
\begin{align*}
\Delta^{-it} A\Delta^{it}B\Om
   & = e^{it\mathscr{H}}Ae^{-it\mathscr{H}}B\Om
\\ & = U(t)AU(-t)B\Om
\\ & = (e^{\beta i tN}(A(e^{-i \beta tN}Be^{i \beta tN}))e^{-i\beta tN})\Om
\\ & = e^{\beta i tN}Ae^{-i \beta tN}B\Om.
\end{align*}
Since $\Om$ is cyclic in $\HS$, and $\HS$ is dense in the weak operator topology of $\calL(H^2(\mathbb{D}))$, it follows that
$ \sigma_{t}(A)= \Delta^{it} A\Delta^{-it} = e^{\beta i tN}Ae^{-i \beta tN}$.
By Theorem \ref{thm:quantum-phase}, for $A = E_B$ with $B\subseteq \R$ a Borel set this may be rewritten in the form
$$ \sigma_{t}(E_B) = e^{i\beta tN}E_B e^{-i\beta tN} = E_{e^{i\beta t}B}.$$
\end{proof}

In conclusion, the thermal time associated with the Gibbs state $\om$ of $N$
recovers the `preferred time' with respect to which it is Gibbs for the unitary group generated by $iN$.

\section{Modular time}\label{sec:modular-time}

Scrutinising the above proof, we note that the unitary group $(\Delta^{it})_{t\in\R}$ implementing thermal time has two key properties:
\begin{enumerate}
\item It acts on a noncommutative $L^2$ space, namely the space of Hilbert-Schmidt operators $\HS$.
\item It is associated with the cyclic vector $\Om= T^{1/2}$, where $T$ is an injective positive operator in $\calL(H)$, in the sense that
$$
\Delta^{it}A = T^{it} A T^{-it} \quad \forall A \in \calL(H^2(\mathbb{D})).
$$
\end{enumerate}
In the proof, the modularity relation
$$
J\Delta^{1/2}(A\Om) = A^{*}\Om \quad \forall A\in\calL(H^2(\mathbb{D}))
$$
with $J:A \mapsto A^{*}$, is a consequence of the above properties of $(\Delta^{it})_{t\in \R}$.

This motivates the following definition.
Let $H$ be a Hilbert space, and let
$\calM\subseteq \calL(H)$ be a von Neumann algebra.
Recall from Section \ref{subsec:NCL2} that if $\tau$ is a trace on $\calM$, then
the corresponding noncommutative space $L^2(\calM,\tau)$ is denoted by $H_{\tau}$.

\begin{definition}
\label{def:modtime}
Let $\tau$ be a trace on $\calM$ and let $T$ be a (possibly unbounded) positive and selfadjoint operator acting in $H$ admitting bounded imaginary powers.
The pair $(\tau,T)$ is said to define a {\em modular time} if the assignment
$$
t \mapsto [A \mapsto T^{it}AT^{-it}]
$$
defines a unitary $C_0$-group of the form $(\Delta^{it})_{t \in\R}$ on $H_{\tau}$, where $\Delta$ is an injective positive selfadjoint operator on $H_\tau$.
\end{definition}

For the theory of bounded imaginary powers we refer to \cite{HNVW3}.

For operators $A\in \calM$ satisfying $\tau(A^\star A)$ we use the standard abuse of notation of identifying
$A$ with its equivalence class modulo $\mathscr{N}_\tau$, which is defined as the left ideal of all $A\in\calM$ satisfying $\tau(A^\star A) = 0$. Denoting the quotient map by $q$, the assignment $A \mapsto T^{it}AT^{-it}$
in the above definition should more rigorously be interpreted as the assignment
$$q(A) \mapsto q(T^{it}AT^{-it}),$$
initially defined for $A\in\calM$ such that $\tau(A^\star A)<\infty$.
We will omit the mapping $q$ in what follows.

The use of the term ``modular" comes from the following fact.

\begin{lemma}\label{lem:modular}
Let $(\tau,T)$ define a modular time, and let $J$ be the antiunitary operator given by the assignment $A \mapsto A^{*}$ on $H_{\tau}$. If for all $A\in H_{\tau}$ the operator $AT^{1/2}$ is well defined as an element of $H_\tau$, then for all $A\in H_\tau$ it belongs to $\Dom(\Delta^{1/2})$ and
$$J\Delta^{1/2}(AT^{1/2}) = A^{*}T^{1/2}.$$
\end{lemma}

Note that if $T$ is bounded, then the assumption that $AT^{1/2}$ be well defined as an element of $H_\tau$ is satisfied. In the setting of Section \ref{sec:NC} the operator $T$ is unbounded and the condition will be checked by hand.

\begin{proof}
By definition, for all $f,g\in \Dom(T^{1/2})$ we have
\begin{align*}
\iprod{J\Delta^{it}(AT^{1/2})f}{g} = \iprod{T^{1/2}T^{-it}A^{*}T^{it}f}{g},
\end{align*}
for all $A \in H_{\tau}$ and $t \in \R$. Since $\Delta$ and $T$ is are injective positive selfadjoint operators, by the edge of the wedge theorem this identity uniquely extends by continuity and analyticity to
\begin{align*}
J\Delta^{1/2}(AT^{1/2}) = A^{*}T^{1/2}.
\end{align*}
The rigorous justification that the elements $AT^{1/2}$ belongs to $\Dom(\Delta^{1/2})$ can be given on the basis of an element belongs to the domain of a densely defined closed operator if and only if it belongs to its weak domain (see \cite[Proposition 10.20]{Nee}).
\end{proof}

The quantum harmonic oscillator of Section \ref{sec:number-phase} corresponds to the special case where
 $H = H^2(\mathbb{D})$, $\calM = \calL(H)$, $\tau$ is the standard trace on $\calM$, and the injective positive operator $T\in\calL(H)$ with $\tr(T)=1$ was given
as $T = e^{-\beta\H}/\tr(e^{-\beta \H})$. In that situation one has $T\in\calM$ and $T^{1/2} \in H_\tau = \HS$.

In the next sections, we construct modular times and POVM on the Borel subsets of $\R$, that satisfy the covariance relation defined below.

\begin{definition}
\label{def:cov}
Let $(\tau,T)$ define a modular time, and $E: \mathcal{B}(\R) \to \mathcal{E}(H)$ define a POVM. We say that $E$ is covariant with respect to the modular time if
$$
T^{it} E_{B} T^{-it} = E_{B+t}
$$
for all Borel sets $B \subseteq \mathbb{R}$ and all $t \in \R$.
\end{definition}

It is important to note that we do not insist that the operators $E_B$ belong to $H_{\tau}$. When they do, the covariance relation can be rewritten as
$$ \Delta^{it}E_B\Delta^{-it} = E_{B+t}. $$
Moreover, we do not insist on having the Hilbert space $H_{\tau}$ arise from Tomita--Takesaki modular theory. Whether or not our constructions could be modified to have these two additional features (and thus fully embed into Tomita--Takesaki modular theory) is left as an open problem (see Section \ref{sec:open}).

\section{Modular time for the free massless relativistic particle}\label{sec:pos-mom}

In this section we address the problem of modular time for the free massless relativistic particle in dimension $1$.
Its Hamiltonian $|D|:=\sqrt{-\Delta}$, where $\Delta$ is the second derivative operator, is the Fourier multiplier on $L^2(\R)$ corresponding to the multiplier $x\mapsto |x|$. Its part in the closed subspace of all $f\in L^2(\R)$
whose Fourier transform is supported on the nonnegative half line $\R_+$ can then be viewed mathematically as the continuous analogue of the number operator $N$. This simple observation suggests that we may try to redo the identification of thermal time more or less {\em mutatis mutandis}.

\paragraph{Modular time as a POVM in $H^2(\mathbb{U})$}

The reasoning of Section \ref{sec:number-phase} suggests that it should be possible to carry out an analogous construction in which
$H= H^2(\D)$ is replaced by $H$, where $$\U := \bigl\{(u,v) \in \R^2: \, u\in\R, \, v>0\bigr\}.$$
The space $H^2(\U)$ is defined as the Hilbert space of all holomorphic functions $f:\U\to \C$ for which
$$ \n f\n_{H^2(\U)} := \sup_{v>0} \n f(\cdot+iv)\n_{L^2(\R)}$$
is finite.
Under convolution with the Poisson kernel for the upper half-plane, this space is isometric to the
closed subspace in $L^2(\R)$ consisting of all functions whose Fourier--Plancherel transform vanishes on the the negative half-line.

If we think of the number operator on $H^2(\mathbb{D})$ as being given by the Fourier multiplier operator on $L^2(\mathbb{T})$ with multiplier $n$ on the nonnegative frequencies, it is natural to consider the operator on $H^2(\U)$ given by the Fourier multiplier operator on $L^2(\R)$ with multiplier $\xi$ on the nonnegative frequences. A positive selfadjoint operator
on $L^2(\R)$ implementing this property is the Poisson operator
$$|D|:= (D^2)^{1/2} = (-\Delta)^{1/2},$$ where $\Delta$ is the Laplace operator (not to be confused with the modular operator, which is also denoted by this symbol), and $D = \frac1i {\rm d}/{\rm d}x$ with domain
\begin{align*}
\Dom(|D|) := \bigl\{f\in H^2(\U): \, |D| f \in L^2(\R)\bigr\}.
\end{align*}
It is not hard to prove that the spectrum of this operator equals $[0,\infty)$.
This operator generates a $C_0$-semigroup of contractions $(P(t))_{t\ge 0}$ on $L^2(\R)$, the so-called Poisson semigroup,
given explicitly as a convolution semigroup
$$ P(t)f = p_t*f ,\quad t\ge 0, \ f\in L^2(\R),$$
where $$p_t(x) = \frac1\pi \frac{t}{t^2+x^2}, \quad t\ge 0, x\in\R,$$
is the {\em Poisson kernel} and the convolution $g*h$, with $g\in L^1(\R)$ and $h\in L^2(\R)$, is defined
by
\begin{equation}\label{eq:convolution} (g*h)(x) := \int_{\R^d} g(x-y)h(y)\ud y.
\end{equation}

The Fourier multiplier with symbol $\one_{[0,\infty)}$ can be expressed in the Borel functional calculus of the selfadjoint
operator $D$ as $\one_{\R_+}(D)$. Let $L_+^2(\R)$ denote the range of this operator, that is,
$$ L_+^2(\R) = \bigl\{f\in L^2(\R): \ \wh f(-\xi) = 0 \ \hbox{for almost all} \ \xi\in\R_+ \bigr\}.$$
Here, $f\mapsto \wh f$ is the Fourier--Plancherel transform, defined for functions $f\in L^1(\R)\cap L^2(\R)$ by
 \begin{equation}\label{eq:FT} \wh f(\xi) := \frac1{\sqrt{2\pi}} \int_{-\infty}^\infty f(x)\exp(-i x\xi)\ud x, \quad \xi\in \R,
 \end{equation}
and extended to an isometry from $L^2(\R)$ onto itself by density.
With the normalisations chosen in \eqref{eq:convolution} and \eqref{eq:FT},
$f\in L^1(\R^d)$ and let $g\in L^1(\R^d)$ or $g\in L^2(\R^d)$. For almost all $\xi\in \R^d$ we have
\begin{align}\label{eq:convolution-FT} \wh{f*g}(\xi) = \sqrt{2\pi} \wh f(\xi) \wh g(\xi).
\end{align}

It is well known that, for all $f\in L_+^2(\R)$, the function $F$ defined by
$$ F(x+iy):= P(y)f(x)$$ belongs to $\U$ with $\n F\n_{H} = \n f\n_2$. In the converse direction, for all $F\in H$ the limit
$$ f:= \lim_{y\downarrow 0} F(\cdot+iy)$$
exists, with convergence in the norm of $L^2(\R)$, and $F$ is recovered from $f$ as $F(x+iy) = P(y)f(x)$.
This correspondence defines a unitary mapping
\begin{align*}
\mathcal{P}: H^{2}(\mathbb{U}) \to L_+^2(\R), \quad  & \mathcal{P} : F \to f,
\end{align*}
whose inverse is given by $f\mapsto F$.

Let $E: \mathscr{B}(\R)\to \mathscr{E}(H^2(\U))$
be uniquely defined by
\begin{align*}
\mathcal{P} E_B F := \one_{[0,\infty)}(D) \one_{B}(X)  \mathcal{P} F
\end{align*}
for Borel sets $B\in \mathscr{B}(\R)$ and functions $F\in H^2(\U)$. Here,
 $\one_{B}(X)$ denotes the multiplication operator on $L^2(\R)$ defined by $\one_{B}(X)f:x \mapsto \one_{B}(x)f(x)$.

We will show that $B \mapsto E_{B}$ is a POVM in $H^2(\mathbb{U})$, and that it is covariant with respect to the flow $S \mapsto e^{it|D|}Se^{-it|D|}$ on the Hilbert space $H_{\tau}$, where $\tau$ is the trace on
$\calL(H^2(\mathbb{U}))$ defined by
$$\tau(A):=  \tr(Ae^{-\beta |D|}).$$

\begin{theorem}\label{thm:thermal-D}
Under the above assumptions, the following assertions hold:
\begin{enumerate}[\rm(1)]
\item\label{it:thermal-D1}
The assignment $B \mapsto E_B$ defines a POVM in $H^2(\mathbb{U})$.
\item\label{it:thermal-D2}
The pair $(\tau,e^{-\beta|D|})$ defines a modular time, in the sense of Definition \ref{def:modtime}.
\item\label{it:thermal-D3}
The POVM $E$ is covariant with respect to this modular time, in the sense of Definition \ref{def:cov}.
\end{enumerate}
\end{theorem}

\begin{proof}
To prove \ref{it:thermal-D1}, we first note that
 $$0\le E_B \le I,$$ that is, every $E_B$ is an effect. To prove this, set $f :=\mathcal{P}F$. Using that $f=\one_{[0,\infty)}(D)f$ and that
$\one_{[0,\infty)}(D) $ and $\one_{B}(X)$ are orthogonal projections, we obtain
\begin{align*}\iprod{E_B F}{F}_{H}
& = \iprod{\one_{[0,\infty)}(D) \one_{B}(X)  f}{f}_{L_+^2(\R)}
\\ & = \iprod{\one_{B}(X)  f}{\one_{[0,\infty)}(D) f}_{L_+^2(\R)}
\\ & = \iprod{\one_{B}(X)  f}{f}_{L^2(\R)}
\\ & = \iprod{\one_{B}(X) f}{\one_B(X) f}_{L^2(\R)}
\end{align*}
and therefore $0\le E_B$; the inequality $E_B\le I$ then follows from
\begin{align*}
\iprod{E_B F}{F}_{H^2(\mathbb{U})} &=
\n \one_B(X) f\n_{L^2(\R)}^2
\le \n f\n_{L^2(\R)}^2  = \n F\n_{H^2(\mathbb{U})}^2 \\
&= \iprod{F}{F}_{H^2(\mathbb{U})}.
\end{align*}

Turning to \ref{it:thermal-D2}, for $A \in \mathcal{L}(H^{2}(\mathbb{U}))$we have that
\begin{align*}
\tr(e^{it|D|}A^{*}Ae^{-it|D|}e^{-\beta|D|}) &= \tr(A^{*}Ae^{-it|D|}e^{-\beta|D|}e^{it|D|})\\
&= \tr(A^{*}Ae^{-\beta|D|}),
\end{align*}
and thus the group defined by $(t \mapsto [A\mapsto e^{it|D|}Ae^{-it|D|}])_{t\in \R}$ is unitary on $H_{\tau}$.

We finally prove \ref{it:thermal-D3}. For all $F,G\in H^2(\U)$, let $f =\mathcal{P}F$ and $g = \mathcal{P}G$. We have,
using \eqref{eq:convolution-FT},
\begin{align*}
\iprod{E_B e^{it|D|}F}{G}
& = \iprod{\one_{[0,\infty)}(D) \one_{B}(X)  \mathcal{P} e^{it|D|} F}{\mathcal{P}G}
\\ & = \iprod{\one_{B}(X) e^{it|D|}\mathcal{P} F}{\one_{[0,\infty)}(D) \mathcal{P}G}
\\ & = \iprod{\one_{B}(X) e^{it|D|}  f}{g}
\\ & = \iprod{\one_{B}(X) e^{itD}  f}{g}
\\ & =  \sqrt{2\pi}\int_0^\infty [\wh{\one_B}* (e^{iu\cdot}\wh f)](u)\ov{\wh g(u)}\ud u
\\ & =  \sqrt{2\pi}\int_0^\infty \int_0^\infty \wh{\one_B}(u-v)e^{itv}\wh f(v)\ov{\wh g(u)}\ud v\ud u.
\end{align*}
At the same time,
\begin{align*}
\iprod{&e^{it|D|}E_{B+t}F}{G}
    =  \iprod{\one_{[0,\infty)}(D) \one_{B+t}(X)  \mathcal{P}F}{\mathcal{P}e^{-it|D|}G}
\\ & =  \iprod{\one_{[0,\infty)}(D) \one_{B+t}(X) f}{e^{-it|D|}g}
\\ & =  \iprod{\one_{B+t}(X)f}{\one_{[0,\infty)}(D)e^{-it|D|}g}
\\ &= \iprod{\one_{B+t}(X) f}{{e^{-itD}g}}
\\ & = \sqrt{2\pi}\int_0^\infty e^{itu}\,\ov{\wh{g}(u)}\int_0^\infty \wh{\one_{B+t}}(u-v)\wh f(v)\ud v\ud u
\\ & = \sqrt{2\pi}\int_0^\infty e^{itu}\,\ov{\wh{g}(u)}\int_0^\infty \wh f(v)\int_\R e^{-i(u-v)z}\one_{B+t}(z)\ud z\ud v\ud u
\\ & = \sqrt{2\pi}\int_0^\infty e^{itu}\,\ov{\wh{g}(u)}\int_0^\infty \wh f(v)\int_\R e^{-i(u-v)(z+t)}\one_{B}(z)\ud z\ud v\ud u
\\ & = \sqrt{2\pi}\int_0^\infty \ov{\wh{g}(u)}\int_0^\infty e^{itv}\wh f(v)\int_\R e^{-i(u-v)z}\one_{B}(z)\ud z\ud v\ud u
\\ & = \sqrt{2\pi}\int_0^\infty \int_0^\infty \wh{\one_{B}}(u-v)e^{itv}\wh f(v)\ov{\wh{g}(u)}\ud v\ud u.
\end{align*}
\end{proof}

\section{Modular time for the noncommutative integral}\label{sec:NC}

In this section we consider a von Neumann algebra $\mathcal{M}$ contained in $\mathcal{L}(L^{2}(\R))$ generated by the Weyl calculus $a \mapsto a(Q,P)$ of a pair of operators $Q,P$ (defined later) for compactly based classical symbols $a$, endowed with the noncommutative integral $\tau$ as a semifinite trace,  that is,
\begin{align*}
\tau(a(Q,P)) &:= {\rm Res}(a(X,D)|D|^{-1}(I-\eta(D))) \\ &= \frac{1}{2} \int _{\R} (a_{0}(x,1)+a_{0}(x,-1))\ud x,
\end{align*}
where $D=\frac{1}{i}{\rm d}/{\rm d}x$, $X:f\mapsto [x\mapsto xf(x)]$,
$\eta \in C_{c}^{\infty}(\R)$ is an arbitrary cut-off function supported on $[-1/2,1/2]$, and $a_{0}$ denotes the principal symbol of $a(X,D)$.
Note that this definition only depends on the symbol $a$ (in fact on its principal part) and not on the choice of operators $Q,P$.
Here ${\rm Res}$ denotes the noncommutative residue (which can be seen as a Dixmier trace). See \cite[Chapters 10,11]{LSZ} for more information on the noncommutative integral, and \cite{Stein} for terminology and fundamental results on pseudo-differential operators and the Weyl calculus.

The expression for $\tau$ is reminiscent of the states $A \mapsto \tr(AT)$ considered in previous sections, but requires the use of $|D|^{-1}$ instead of $\exp(-\beta|D|)$. Consequently, the translation group generated by $iD$ will be replaced by the $C_{0}$ group of imaginary powers $(|D|^{it})_{t \in \R}$, i.e., the group generated by  $iP$ where $$P:= \ln |D|.$$

We show that the pair $(\tau,|D|^{-1})$  defines a modular time $(\Delta^{it})_{t \in \R}$ in the sense of Definition \ref{def:modtime}.
We then construct a POVM that is covariant with respect to this modular time in the sense of Definition \ref{def:cov}.
This will be constructed using generalised pseudo-differential operators of the form $a(Q,P)$ (in the sense of the Weyl calculus defined in \cite{vNP-JOT}), where $Q$ is a suitable replacement for the standard position operator, adapted to our modified momentum operator $P$.

\begin{definition}
Let $(S(t))_{t \in \R}$ be the unitary $C_0$-group on $L^{2}(\R)$ defined by
\begin{align*}
S(t): g \mapsto e^{t/2} \mathcal{F}^{-1}( \xi \mapsto \wh{g}(e^{t}\xi)),
\end{align*}
for all $t \in \R$. We denote its generator by $Q$.
\end{definition}

\begin{lemma}
The group $(S(t))_{t\geq 0}$ satisfies the following Weyl relations:
\begin{equation}
\label{eq:weyl}
e^{isP}S(t) =  e^{-ist} S(t)e^{isP}, \quad \forall s,t\in\R.\end{equation}
\end{lemma}

\begin{proof}
For $f \in L^{2}(\R)$ and $\xi \in \R$ we have
\begin{align*}\  \wh{e^{isP}S(t)f}(\xi) & = |\xi|^{is}\wh {S(t)f}(\xi) = e^{t/2}|\xi|^{is}\wh{f}(e^t \xi)
\intertext{as well as}
\wh{S(t)e^{isP}f}(\xi) &= e^{t/2}|e^{t}\xi|^{is}\wh f(e^t \xi)
= e^{ist}e^{t/2}|\xi|^{is}\wh{f}(e^t \xi).
\end{align*}
\end{proof}

\paragraph{Modular time as a POVM}

We will use the group functional calculi for $P$ and $Q$, defined by
\begin{align*}
f(P) &:= \frac{1}{\sqrt{2\pi}} \int \limits _{\R} \wh{f}(\xi) \exp(i\xi P) \ud \xi,\\
g(Q) &:= \frac{1}{\sqrt{2\pi}} \int \limits _{\R} \wh{g}(\xi) \exp(i\xi Q) \ud \xi,\end{align*}
for $\wh{f},\wh{g} \in L^{1}(\R)$, as well as the Weyl functional calculus of \cite{vNP-JOT},  defined
for $a \in S^{0}(\R^{2})$ by
$$
a(Q,P) = \frac{1}{2\pi} \int \limits _{\R^{2}} \widehat{a}(u,v) \exp(\frac{i}{2} u\cdot v)\exp(iuQ)\exp(ivP)\ud u\ud v.
$$

\begin{theorem}\label{thm:thermal-Dixmier}
For any Borel set $B\in \mathscr{B}(\R)$ define $E_B\in \mathcal{A}$ by
$$ E_B := \one_{\R_+}(P)\one_B(Q)\one_{\R_+}(P),$$
the following assertions hold:
\begin{enumerate}[\rm(1)]
\item\label{it:Dixmier1}
The assignment $B \mapsto E_B$ defines a POVM on the range space $\ran(\one_{\R_+}(P))$.
\item\label{it:Dixmier2}
The pair $(\tau,|D|^{-1})$ defines a modular time, in the sense of Definition \ref{def:modtime}.
\item\label{it:Dixmier3}
The POVM $E$ is covariant with respect to this modular time, in the sense of Definition \ref{def:cov}.
\end{enumerate}
\end{theorem}

Observe that  $\ran(\one_{\R_+}(P))$ naturally takes over the role of the Hardy spaces used in the Sections \ref{sec:number-phase} and \ref{sec:pos-mom}.

\begin{proof}
To prove \ref{it:Dixmier1}, we first note that
\begin{align*}
\one_B(Q)^{2} & =\frac{1}{\sqrt{2\pi}} \int \limits _{\R}  \int \limits _{\R} \wh{\one_{B}}(\eta)\wh{\one_{B}}(\xi) \exp(i(\xi+\eta)Q)\ud \xi \ud\eta
\\ & = \frac{1}{\sqrt{2\pi}} \int \limits _{\R} \wh{\one_{B}}\star\wh{\one_{B}}(\xi) \exp(i\xi Q)\ud \xi
= \one_{B}(Q).\end{align*}
Moreover, for all $\xi \in \R$ and $f,g \in L^{2}(\R)$ we have
\begin{align*}
\iprod{S(\xi)f}{g} &= \iprod{f}{S(-\xi)g},
\end{align*}
by a change of variables.
and thus $\iprod{\one_{B}(Q)f}{g} = \iprod{f}{\one_{B}(Q)g}$, and
\begin{align*}
\iprod{E_{B} f}{f} & = \iprod{\one_B(Q)\one_{\R_+}(P)f}{\one_{\R_+}(P)f}
\\ & = \iprod{\one_B(Q)^{2}\one_{\R_+}(P)f}{\one_{\R_+}(P)f} \\
& = \|\one_B(Q)\one_{\R_+}(P)f\|^{2} \geq 0.
\end{align*}
Moreover $\|\one_B(Q)\one_{\R_+}(P)f\|^{2} =  \iprod{\one_B(Q)\one_{\R_+}(P)f}{\one_{\R_+}(P)f} \leq \|\one_B(Q)\one_{\R_+}(P)f\|\|f\|$,
which implies $E_{B} \geq 1$.
For Borel sets $(B_{n})_{n \in \N}$, we also have
\begin{align*}
\iprod{\one_{\bigcup_{n \in \N}B_{n}}(Q) f}{f} &= \frac{1}{\sqrt{2\pi}} \int \limits _{\R} \sum _{n \in \N} \wh{\one_{B_{n}}}(\xi) S(\xi)\ud \xi \\ &= \sum _{n \in \N} \iprod{\one_{B_{n}}(Q) f}{f},
\end{align*}
and $\iprod{\one_{\R}(Q)f}{f} = \iprod{S(0)f}{f} = \iprod{f}{f}$. This completes the proof of (1).

\smallskip
For \ref{it:Dixmier2}, we consider $a(Q,P) \in H_{\tau}$. Using \eqref{eq:weyl}, we first note that, for $t \in \R$,
\begin{align*}
&|D|^{it}a(Q,P)|D|^{-it} \\  &= \frac{1}{2\pi} \int \limits _{\R^{2}} \widehat{a}(u,v) e^{\frac{i}{2} u\cdot v}\exp(itP)\exp(iuQ)\exp(i(v-t)P)\ud u\ud v \\
& = \frac{1}{2\pi} \int \limits _{\R^{2}} \widehat{a}(u,v) e^{-itu}e^{\frac{i}{2} u\cdot v}\exp(iuQ)\exp(ivP)\ud u\ud v
\\ & = a_{t}(Q,P),
\end{align*}
where $a_{t}:(x,\xi)\mapsto a(x+t,\xi)$. This implies that
$t\mapsto |D|^{it}a(Q,P)|D|^{-it}$ is a group, and that
\begin{align*}
\|T^{it}a(Q,P)T^{-it}\|_{H_{\tau}} ^{2} &= \frac{1}{2} \int _{\R}
(a_{0}^{2}(x+t,1)+a_{0}^{2}(x+t,-1))\ud x \\ &= \|a(Q,P)\|_{H_{\tau}} ^{2}.
\end{align*}
To verify the abstract condition in Lemma \ref{lem:modular} that $AT^{1/2}\in H_\tau$ be well defined and an element of $H_\tau$, we argue as follows.
Noting that $a(Q,P)$ and $a(Q,P)(I-\eta(\exp(P))$ are both representatives of the same equivalence class in $H_{\tau}$, we have that $a(Q,P)\exp(-\frac{1}{2}P)$ agrees on a dense set with $a(Q,P)(I-\eta(\exp(P))\exp(-\frac{1}{2}P) \in H_{\tau}$.

\smallskip
We argue similarly for \ref{it:Dixmier3}, using the commutation relation
\eqref{eq:weyl} to obtain that, for all $t \in \R$ and $B \in \mathcal{B}(\R)$,
\begin{align*}
|D|^{it}& E_{B}|D|^{-it} \\ &= \frac{1}{\sqrt{2\pi}} \one_{\R_{+}}(P)  \int \limits _{\R} \wh{\one_{B}}(\xi) e^{itP}e^{i\xi Q} e^{-itP}\one_{\R_{+}}(P) \ud \xi\\
& = \frac{1}{\sqrt{2\pi}} \one_{\R_{+}}(P)  \int \limits _{\R} \wh{\one_{B}}(\xi) e^{-it\xi}e^{i\xi Q} \one_{\R_{+}}(P)\ud \xi = E_{B+t}.
\end{align*}
\end{proof}

\section{Open problems}\label{sec:open}

We believe that the following problems are worth investigating:

\begin{itemize}
 \item Do the results of Section \ref{sec:pos-mom} extend to the free relativistic particle with mass $m$?
\end{itemize}
The obstruction to overcome is that the algebra underlying our computations no longer works for the Hamiltonian
$\sqrt{m^2-\Delta}$ with $\Delta$ (the Laplace operator). In all likelihood, techniques involving Dirac operators have to be used to handle this case.

\begin{itemize}
 \item Embed Theorems \ref{thm:thermal-D} and \ref{thm:thermal-Dixmier} into the Tomita--Takesaki framework.
\end{itemize}

As was already pointed out, there are two obstructions to be overcome: (i) both theorems involve a non-normal semifinite trace; (ii) the operators $E_B$ in the POVM's constructed in these theorems generally do not belong to the von Neumann algebra $\calM$.
used in these theorems.

\begin{itemize}
 \item Identify a physical system whose equilibrium state is described by the non\-commut\-ative integral.
\end{itemize}

This question is speculative; we have included the treatment of the noncommutative integral simply because it invited itself on mathematical grounds.

\begin{itemize}
\item Interpret the ``observation" of time given in  \cite{Tanaka} as a POVM associated with a modular/thermal time.
\end{itemize}
The paper  \cite{Tanaka} provides a completely different approach to the thermodynamical origin of time compared to the Connes-Rovelli theory. It identifies an information geometric relationship between time and temperature (derived in a simple model from classical mechanics), and suggests using this relationship (stating that a curvature is constant equal to -1) as a way to ``observe" time.
\begin{itemize}
\item Interpret the time POVMs given here in the Page-Wootters formalism
(introduced in \cite{PW} and corrected in \cite{getall}), i.e., as internal quantum degrees of freedom).
\end{itemize}
This could well use the relativisation procedure from \cite{LovMiy}.

\bibliographystyle{plain}
\bibliography{literature-time}

\end{document}